\documentclass[lettersize,journal]{IEEEtran}
\usepackage{amsmath,amsfonts}
\usepackage{algorithmic}
\usepackage{algorithm}
\usepackage{array}
\usepackage{textcomp}
\usepackage{stfloats}
\usepackage{url}
\usepackage{verbatim}
\usepackage{graphicx}
\usepackage{caption}
\usepackage{float} 
\usepackage{subcaption}
\usepackage{cite}

\usepackage{bm}
\usepackage{color}
\usepackage{xcolor}
\usepackage{amsthm}

\newtheorem{lemma}{\underline{\textit{Lemma}}}
\newtheorem{proposition}{\underline{\textit{Proposition}}}
\theoremstyle{definition}
\newtheorem{definition}{\underline{\textit{Definition}}}
\newtheorem{remark}{\underline{\textit{Remark}}}
\usepackage{amssymb}
\begin{document}

\title{Joint Frequency-Space Sparse Reconstruction for DOA Estimation under Coherent Sources and Amplitude-Phase Errors}
\author{Yutong Chen, Cong Zhou, Changsheng You,~\IEEEmembership{Member,~IEEE}, Shuo Shi
\thanks{Yutong Chen, Cong Zhou and Shuo Shi are with the School of Electronics and Information Engineering, Harbin Institute of Technology, Harbin, 150001, China (e-mail:  chenyutong@stu.hit.edu.cn, zhoucong@stu.hit.edu.cn, crcss@hit.edu.cn).}
\thanks{Changsheng You is with the Department of Electronic and Electrical Engineering, Southern University of Science and Technology, Shenzhen 518055, China. (e-mails: youcs@sustech.edu.cn).}
\thanks{\emph{Corresponding authors: Shuo Shi and Cong Zhou.}}}

\maketitle

\begin{abstract}
In this letter, we propose a joint frequency-space sparse reconstruction method for direction-of-arrival (DOA) estimation, which effectively addresses the issues arising from the existence of coherent sources and array amplitude-phase errors.
Specifically, by using an auxiliary source with known angles, we first construct the real steering vectors (RSVs) based on the spectral peaks of received signals in the frequency domain, which serve as a complete basis matrix for compensation for amplitude-phase errors.
Then, we leverage the spectral sparsity of snapshot data in the frequency domain and the spatial sparsity of incident directions to perform the DOA estimation according to the sparse reconstruction method.
The proposed method does not require iterative optimization, hence exhibiting low computational complexity. 
Numerical results demonstrate that the proposed DOA estimation method achieves higher estimation accuracy for coherent sources as compared to various benchmark schemes.
\end{abstract}


\section{Introduction}
Direction-of-arrival (DOA) estimation has wide applications in wireless communications, radar systems, and localization \cite{molaei2024comprehensive}. 
Classical DOA estimation methods such as multiple signal classification (MUSIC) \cite{zhao2021phd}, maximum likelihood (ML) \cite{yang2023robust}, and subspace fitting (SF) \cite{meng2019robust} have been extensively studied and widely applied due to their high resolution and theoretical robustness under ideal conditions (e.g., non-coherent signals and perfect amplitude-phase manipulation). 
However, in practical scenarios, the above methods suffer from performance degradation due to inevitable array imperfections such as amplitude-phase errors, mutual coupling, and antenna position mismatches \cite{9397386}.

To address these issues, various calibration-based methods have been proposed, which can be divided into active and self-calibration methods.
Active-calibration methods leverage auxiliary sources with known incident directions to estimate and compensate for array imperfections. 
For example, the authors in \cite{yang2022theory} proposed an efficient method using real steering vectors (RSVs) constructed from the intermediate frequency (IF) sampled signal, accounting for antenna and channel errors. This method has been shown to achieve better performance, especially for directional antennas, while reducing computational complexity via discrete Fourier transform (DFT) techniques.
Moreover, although self-calibration methods avoid auxiliary sources, they usually have issues associated with oversimplified assumptions.
For example, the authors in \cite{wang2023self} proposed an efficient method to jointly estimate the array errors and DOAs of the sources using an alternative minimization (AM) algorithm, which can improve the robustness of DOA estimation and reduce computational complexity.
In addition, by minimizing the second moment of the phase error via blind signal separation, a non-iterative self-calibration method for large-scale planar arrays was proposed in \cite{dai2020gain}.
It achieves accurate two dimensional (2D) DOA and amplitude-phase error estimation with low computational complexity and requires only one signal source.
The above works assume that the incident sources are mutually independent.
However, correlated or coherent signals naturally exist due to the multi-path effect in wireless propagation channels.
Hence, when the above methods are applied, the rank of the covariance matrix will be deficient, resulting in severe estimation failure. 

To solve the issues of coherent signals, the authors in \cite{pan2022simplified} proposed a method based on spatial smoothing and Vandermonde structures. 
However, the spatial smoothing method inevitably reduces the effective aperture and resolution, resulting in degraded estimation performance. 
In addition, a structured tensor reconstruction method was proposed in \cite{zheng2022structured} without using spatial smoothing, for which they enforce a tensorial Hermitian Toeplitz mapping on the rank-deficient covariance tensor.
However, this method incurs extremely high computational complexity.
Recently, DOA estimation methods based on sparse reconstruction and compressive sensing are emerging, which can avoid rank deficiency and achieve good performance under coherent sources.
For instance, the authors in \cite{9127154} proposed a method that enhances signal sparsity with a designed weight vector and achieves robust DOA estimation through $ \ell_1 $-norm optimization.
However, these methods still exhibit performance degradation under amplitude-phase errors.

In view of the above works, there still exist several limitations under coherent sources and amplitude-phase errors. To address these issues, in this letter, by leveraging the dual-sparsity in the frequency and spatial domains, we propose a real steering vector-sparse reconstruction (RSV-SR) based method for DOA estimation under narrow-band coherent sources and amplitude-phase errors.
Specifically, we first construct the RSVs using an auxiliary source to eliminate the amplitude-phase errors.
Next, we transform the received data into the frequency domain, and extract the peak points by exploiting the sparsity in the frequency domain.
Finally, we estimate the DOA via sparse reconstruction in the spatial domain. 
Numerical results demonstrate superior accuracy and resolution compared over various benchmarks under low signal-to-noise ratio (SNR) and/or a limited number of snapshots.

\vspace{-8pt}

\section{System Model}
\underline{\bf{Array model with amplitude-phase errors:}} The base station (BS) is assumed to be equipped with a uniform linear array (ULA) comprising $M$ omni-directional antennas, where the inter-antenna spacing is denoted by $d$. 
Different from existing works that assume non-coherent signals and perfect amplitude-phase measuring, we consider a practical case where there exist amplitude-phase errors typically arising from amplifier gain inconsistencies, which are independent of the DOAs of sources.

\underline{\bf{Incident signal model:}} There are $J$ narrow-band signals with zero means impinging on the ULA and the $J$ signal sources are assumed to be located in the far-field region. 
In particular, the $j$-th complex envelope signal after frequency down-conversion is denoted by $s_j(t), \, j \in \mathcal{J} \triangleq \{ 1,\cdots, J\}$, which has an IF center frequency denoted by $\omega_j$.
In addition, the $J$ signals are partially correlated, i.e., $\mathbb{E}\{s_i(t)s_j^*(t)\} \neq 0, i\neq j$. 
The correlation coefficient of $s_i(t),s_j(t),i\neq j$ is defined by 
\begin{equation}
\rho_{s_i s_j} = \frac{\mathbb{E}\left[ s_i(t)s_j^\ast(t)\right]}{\sqrt{\mathbb{E}\left[ |s_i(t)|^2\right]\mathbb{E}\left[ |s_j(t)|^2\right]}}.
\end{equation}
In this letter, we consider a more challenging case where the incident signals are fully coherent, i.e., $\rho_{s_i s_j}=1$. 
Specifically, we assume that the incident signals share the same angular frequency, i.e., $\omega_1=\omega_2=\cdots=\omega_J \triangleq {\omega_{\rm s}}$.

\begin{figure}[bt]
    \centering
    \includegraphics[width=1\linewidth]{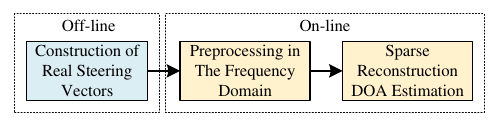}
    \caption{The framework of proposed DOA estimation method}
    \label{fig00:enter-label}
    \vspace{-10pt}
\end{figure}

\underline{\bf{Received signal model:}} The received signal ${\bm{x}}(t) \in \mathbb{C}^{M \times 1}$ at time $t$ is given by \cite{zhou2025mixed}
\begin{equation}
    {\bm{x}}(t) = {\mathbf{\Gamma A}}{\bm{s}}(t) + {\bm{n}}(t).  
\end{equation}
Herein, ${\bf{\Gamma}} = {\rm{diag}}[1,{g_2}{e^{\jmath{\varphi _2}}}, \cdots ,{g_M}{e^{\jmath{\varphi _M}}}]$ denotes the amplitude-phase error matrix, where ${g_m}{e^{\jmath{\varphi _m}}},m \in \mathcal{M} \triangleq \{ 1,\cdots, M\}$ represents the amplitude-phase error of the $m$-th antenna, with the first antenna serving as the reference. 
${g_m}\in \mathbb{R}$ and ${\varphi _m}\in \mathbb{R}$ denote the amplitude and phase errors, respectively. 
In addition, ${\bm{s}}(t) = [{s}_1(t), {s}_2(t), \cdots, {s}_J(t)]^T$ represents the signal vector. 
${\mathbf{A}} = [{\bm{a}}({\theta_1}), {\bm{a}}({\theta_2}), \cdots, {\bm{a}}({\theta_J})] $ denotes the array steering matrix (ASM), where ${\bm{a}}({\theta _j})$ is the array response vector of the $j$-th source, given by
\begin{equation}
    {\bm{a}}({\theta_j}) = \begin{bmatrix}1,e^{\frac{\jmath2\pi d\sin {\theta_j}}{\lambda_{\rm c}} } , \cdots ,e^{\frac{\jmath2\pi(M - 1)d\sin {\theta_j}}{\lambda_{\rm c}} } \end{bmatrix}^{T}.
\end{equation}
Herein, $\theta_j$ denotes the DOA of the $j$-th source, and $\lambda_{\rm c}$ represents the wavelength of the carrier frequency.
Moreover, ${\bm{n}}(t) = {[{n_1}(t),{n_2}(t), \cdots ,{n_M}(t)]^{T}}$ is the received noise vector, where ${\bm{n}}(t) \sim \mathcal{CN}(\bm{0}, \sigma_n^2\bf{I})$ with $\sigma_n^2$ denoting the noise power. 
For simplicity, the array steering matrix can be recast as
${\bf{B}} = {\bf{\Gamma A}} = [{\bm{b}} (\theta_1 ), \cdots, {\bm{b}} (\theta_J )]$, where ${\bm{b}} (\theta_j) = {\bf{\Gamma}}{\bm{a}}(\theta_j)$. Hence, the received signal can be rewritten as
\begin{equation}
    {\bm{x}}(t) = {\bf{B}}{\bm{s}}(t) + {\bm{n}}(t).
    \label{e4}
\end{equation}

\vspace{-8pt}

\section{The Proposed DOA Estimation Method}

In this section, we first discuss the main challenges when estimating the DOA under coherent sources and amplitude-phase errors. 
Specifically, amplitude-phase errors can cause the model mismatch, breaking the Toeplitz structure of the covariance matrix and the subspace orthogonality. For self-calibration methods, the joint optimization of amplitude-phase errors and DOAs leads to a non-convex optimization problem, which is prone to convergence to a local point.
In addition, coherent sources result in rank deficiency of the covariance matrix. Decorrelation techniques such as spatial smoothing incur array aperture loss, reduce degrees of freedom (DOF), and consequently degrade the resolution of DOA estimation.

To address the above issues, we propose an RSV-SR method for DOA estimation, which consists of three phases as illustrated in Fig. \ref{fig00:enter-label}. 
In Phase 1, before DOA estimation, we construct the RSVs using an auxiliary source, which is termed as off-line processing.
The RSVs can be considered as a complete basis matrix that incorporates amplitude-phase errors.
Then, during the on-line processing, we first obtain the Fourier transform of the received signal, and then estimate the sources' DOAs based on the sparse reconstruction method, which can effectively address the rank-deficiency issues arising from coherent sources.
The main definition in this section is given as follows.
\begin{definition}
   \rm{Given a received signal $\bm{x}(t) \in \mathbb{C}^{M\times 1}$ with $L$ snapshots, the measured signal at the $m$-th antenna is denoted by $\bm{x}_m = \big[[\bm{x}(1)]_m, \cdots, [\bm{x}(L)]_m\big]^T \in \mathbb{C}^{L \times 1}$, and the corresponding $L$-point DFT is defined by
   \begin{equation}
       \mathbf{X}_m \triangleq {\rm{DFT}}_L(\bm{x}_m),
   \end{equation}
   where ${\rm{DFT}}_L(\cdot)$ denotes performing an $L$-point DFT.
   Then, the Fourier transform of the received signal ${\bm{x}}(t)$ at frequency $\omega$ is defined by 
    \begin{equation}
       \mathbf{X}(\omega) \!\triangleq\! \big[{\bf{X}}_1(\omega),\mathbf{X}_2(\omega), \cdots\!, \mathbf{X}_M(\omega)\big]^T\!, \omega = 1,2,\cdots,L.
   \end{equation}
   }  
\end{definition}
\vspace{-8pt}
\subsection{Construction of RSVs}
\label{sec:RSV construction}
Before DOA estimation, we sample the spacial domain into $N$ angles denoted by ${\bar\theta _n} = -\frac{\pi}{2} + \frac{n\pi}{N}, n \in \mathcal{N} \triangleq \{1,2, \cdots ,N\} $. 
Then, we place a narrow-band auxiliary source and let it sweep across the discrete angles from ${\bar \theta _1}$ to ${\bar \theta _N}$.
Specifically, the received signal at ${\bar \theta _n}$ is given by
\begin{equation}
    {\bm{x}}(t,{\bar \theta _n}) = {\bm{b}} ({\bar \theta _n}){s_0}(t) + {\bm{n}}(t),
\end{equation}
where ${s_0}(t)$ denotes the transmit signal of the auxiliary source with an angular frequency $\omega_0$. Before obtaining the DFT of the received signal ${\bm{x}}(t,{\bar \theta _n})$, we first introduce the following lemma.
\begin{lemma}
    \rm{Given the received signal ${\bm{x}} ({t}) = {\bf{B}}{\bm{s}}({t})+{\bm{n}}({t})$ with $L$ snapshots, the DFT of the received signal $\bm{x}(t)$ at frequency $\omega$ is given by
    \begin{equation}
        {\bf{X}} ({\omega}) = {\bf{B}}{\bf{S}}({\omega})+{\bf{N}}({\omega}),
    \end{equation}
    where ${\bf{S}}({\omega})$ and ${\bf{N}}({\omega})$ denote the DFT of the transmit signal $\bm{s}(t)$ and noise signal $\bm{n}(t)$ at frequency $\omega$, respectively.
    }
    \label{lemma:linear DFT}
\end{lemma}
\begin{proof}
    Please refer to Appendix A.
\end{proof}

From Lemma \ref{lemma:linear DFT}, the DFT of the received signal ${\bm{x}}(t,{\bar \theta _n})$ is given by
\begin{equation}
    {\bf{X}}(\omega,{\bar \theta _n}) = {\bm{b}} ({\bar{\theta}_n}){\rm{S}_0}(\omega) + {\bf{N}}(\omega),
\end{equation}
where ${\rm{S}_0}(\omega)$ denotes the DFT of ${s_0}(t)$. Since ${s_0}(t)$ is a narrow-band signal with a central frequency $\omega_0$, the transformed ${\bf{X}}(\omega,{\bar \theta _n})$ has peaks at the frequency point $\omega = \omega_0$.
By ignoring the noise signal, the DFT of the array received signal can be approximated as 
\begin{equation}
\widetilde {\bf{X}} (\omega ,{\bar \theta _n}) = \left\{ \begin{array}{l}
{\bm{b}} ({\bar \theta _n}){\rm{S}_0}({\omega _0})\;,\;\;\omega  = {\omega _0}\\
0\;,\quad \quad \quad \quad \;\;\;\,\,\omega  \ne {\omega _0}
\end{array} \right. .
\label{e01}
\end{equation}

By collecting the frequency peaks of the DFT signals at the $N$ discrete points ${\bar \theta _n}\;,n \in \mathcal{N}$, the RSV matrix ${\bf{\Psi }} \in \mathbb{C}^{M \times N}$ is constructed by
\begin{equation}
    \begin{aligned}
        {\bf{\Psi }}  &= \left[ {\widetilde {\bf{X}} ({\omega _0},{{\bar \theta }_1}),\widetilde {\bf{X}} ({\omega _0},{{\bar \theta }_2}), \cdots ,\widetilde {\bf{X}} ({\omega _0},{{\bar \theta }_N})} \right]  \\&=\left[ { {\bm{b}} ({{\bar \theta }_1}){\rm{S}_0}({\omega _0}), \cdots ,{\bm{b}} ({{\bar \theta }_N}){\rm{S}_0}({\omega _0})} \right]  \\&=
        {\rm{S}_0}({\omega _0})\left[ {{\bm{b}} ({{\bar \theta }_1}), \cdots ,{\bm{b}} ({{\bar \theta }_N})} \right]. 
    \end{aligned}
\end{equation}
By setting ${[ {{\bm{a}}({{\bar \theta }_n})} ] _1} = 1,\;{g_1}{e^{\jmath{\varphi _1}}} = 1$, we have ${[ { {\bm{b}} ({{\bar \theta }_n})} ] _1} = 1$.
Hence, ${\bf{\Psi }}$ can be normalized as 
\begin{equation}
    \begin{aligned}
        {\bf{\Psi }}_{\rm norm}  &= \left[ {\frac{{{{\boldsymbol{\psi }}_1}}}{{{{ [{{{\boldsymbol{\psi }}_1}}] }_1}}},\frac{{{{\boldsymbol{\psi }}_2}}}{{{{  [{{{\boldsymbol{\psi }}_2}}]  }_1}}}, \cdots ,\frac{{{{\boldsymbol{\psi }}_N}}}{{{{ [{{{\boldsymbol{\psi }}_N}}]  }_1}}}} \right] \\
        &= 
        \left[ {{\bm{b}} ({{\bar \theta }_1}), \cdots ,{\bm{b}} ({{\bar \theta }_N})} \right],
    \end{aligned}
\end{equation}
where ${{\boldsymbol{\psi}}_n}$ represents the $ n $-th column of the RSV matrix ${\bf{\Psi }}, n\in \mathcal{N}$.

From the above, it is shown that the normalized RSV matrix ${\bf{\Psi }}_{\rm norm}$ is the array manifold formed by the steering vectors corresponding to the sampled angles $\{ {\bar \theta _1},{\bar \theta _2}, \ldots ,{\bar \theta _N}\} $.
Hence, it can be considered as a complete basis matrix in the sparse reconstruction theory, which can eliminate the effect of amplitude-phase errors.
In addition, $\widetilde {\bf{X}} ({\omega _0},{\bar \theta _n})$ is directly obtained from the actual sampled signal, which includes amplitude-phase errors and mutual coupling.
Hence, the RSVs can be used to mitigate the effects of amplitude-phase errors.

\subsection{Preprocessing in the Frequency Domain}
To ensure spatial sparsity, the number of incident signals is assumed to be much smaller than the number of columns in the RSV matrix ${\bf{\Psi }}_{\rm norm}$.
After obtaining the RSV matrix in Section \ref{sec:RSV construction}, based on Lemma \ref{lemma:linear DFT}, we first perform DFT on the received signal in \eqref{e4}, given by
\begin{equation}
    {\bf{X}} ({\omega}) = {\bf{B}}{\bf{S}}({\omega})+{\bf{N}}({\omega}).
    \label{eq:DFT of received signal}
\end{equation}
Similar to the construction of the RSV matrix, we extract the spectral peak in \eqref{eq:DFT of received signal} while neglecting noise, given by\footnote{For the general case $\omega_1 \neq \omega_2\neq\cdots\neq \omega_J$, our proposed method can be directly extended, which is demonstrated in Proposition \ref{proposition1}.}
\begin{equation}
    \widetilde {\bf{X}} ({\omega _{\rm s}}) = {\bf{B}}{\bf{S}}({\omega _{\rm s}}) .
\end{equation}

\begin{remark}
    In the considered scenario with co-frequency sources, incident signals become fully coherent, making the covariance matrix of sources rank-deficiency.
    As such, conventional subspace-based methods such as that in \cite{yang2022theory} cannot be applied, since the signal subspace is not completely orthogonal to the noise subspace.
    Therefore, the sparse reconstruction method is adopted in this letter to solve the rank-deficiency problem as shown in Section \ref{sec:Sparse Reconstruction DOA Estimation}.
\end{remark}

\subsection{Sparse Reconstruction DOA Estimation}
\label{sec:Sparse Reconstruction DOA Estimation}
According to sparse reconstruction theory, by using the RSV matrix ${\bf{\Psi }}_{\rm norm} $, $\widetilde {\bf{X}} ({\omega _{\rm s}})\triangleq \widetilde {\bf{X}} \in \mathbb{C}^{M\times 1}$ can be sparsely expressed by
\begin{equation}
    \begin{aligned}
        \widetilde {\bf{X}}  &= {\bf{\Psi }}_{\rm norm} {\bf{Z}} {\bf{S}}({\omega _{\rm s}}) \\
        &= {\bf{\Psi }}_{\rm norm}\widehat {\bf{S}}({\omega _{\rm s}}),
    \end{aligned}
\label{e11}
\end{equation}
where ${\mathbf{Z}} \in \mathbb{R}^{N \times J}$ is the sparse coefficient matrix, $\widehat {\bf{S}}({\omega _{\rm s}}) \triangleq \widehat {\bf{S}} \in \mathbb{C}^{N\times 1} $ is the spatial sparse vector.
In particular, the entries in ${\mathbf{Z}}$ and $\widehat {\bf{S}}$ are given by
\begin{equation}
    [{\bf{Z}}]_{n,j} = \left\{ \begin{array}{l}
    1,\;\; {\bar \theta _n} = {\theta _j}, \forall j \in \mathcal{J}, n \in \mathcal{N}\\
     0,\;\;{\bar \theta _n} \ne {\theta _j},\forall j \in \mathcal{J}, n \in \mathcal{N}
\end{array} \right. ,
\end{equation}
and
\begin{equation}
    \left\{ \begin{aligned}
        \big[\widehat{\bf{S}}\big]_n \ne 0,\;\;{\bar \theta _n} = {\theta _j},\forall j \in \mathcal{J}, n \in \mathcal{N}\\
        \big[\widehat{\bf{S}}\big]_n = 0,\;\;{\bar \theta _n} \ne {\theta _j},\forall j \in \mathcal{J}, n \in \mathcal{N}
    \end{aligned} \right. .
    \label{e12}
\end{equation}
Then, based on the sparsity of ${\widehat{\bf{S}}}$, we formulate the sparse reconstruction problem as follows.
\begin{equation}
(P1)~~\arg \;\mathop {\min }\limits_{\widehat {\bf{S}}} \;\left\| {\widetilde {\bf{X}}  -  {\bf{\Psi }}_{\rm norm}\widehat {\bf{S}}} \right\|_{\mathop{\rm F}\nolimits} ^2 + \mu {\left\| {\widehat {\bf{S}}} \right\|_0}\;,
\label{e013}
\end{equation}
where $||\cdot|{|_0}$ denotes the $ \ell_0 $-norm, which represents the number of non-zero elements in $\widehat {\bf{S}}$.
In addition, $\mu $ is the penalty factor, which is used to adjust the sparsity of the optimization variable. 
Since (P1) is a non-convex optimization problem, we slack the $ \ell_0 $-norm by the $ \ell_1 $-norm, given by
\begin{equation}
(P2)~~\arg \;\mathop {\min }\limits_{\widehat {\bf{S}}} \;\left\| {\widetilde {\bf{X}}  -  {\bf{\Psi }}_{\rm norm}\widehat {\bf{S}}} \right\|_{\mathop{\rm F}\nolimits} ^2 + \mu {\left\| {\widehat {\bf{S}}} \right\|_1}\;.
\label{e13}
\end{equation}

It is observed that (P2) is a convex optimization problem, which can be effectively solved by CVX.
By denoting the optimal solution to (P2) as $\widehat {\bf{S}}^\ast$, the estimated indices are given by
\begin{equation}
    \widehat{n}_j =~ ^J\!\arg\max~ \widehat{\bf{S}}^\ast, j = 1,2,\cdots,J.
\end{equation}
Then, the estimated DOAs of the $J$ sources are given by 
\begin{equation}
    \widehat{\theta}_j = \bar{\theta}_{\hat{n}_{j}}, \forall j \in \mathcal{J}.
\end{equation}
Next, we extend the proposed method to the general case where $\omega_1 \neq \omega_2\neq\cdots\neq \omega_J$ as follows.

\begin{proposition}
\emph{For the general case $\omega_1 \neq \omega_2\neq\cdots\neq \omega_J$, given the DFT spectrum $\widetilde{\bf{X}}(\omega)$, there exist $J$ peaks in the DFT spectrum denoted by 
$\widetilde{\bf{X}}(\omega_j), \forall j \in \mathcal{J}$.
Then, let $\widetilde{\bf{X}} = \widetilde{\bf{X}}(\omega_1)+\widetilde{\bf{X}}(\omega_2)+\cdots+\widetilde{\bf{X}}(\omega_J)$, the DOA estimation problem can still be solved using (\ref{e13}).}
\label{proposition1}
\end{proposition}
\begin{proof}
From (\ref{e11}), we have
\begin{equation}
\widetilde {\bf{X}} ({\omega _j}) = {\bf{\Psi }}_{\rm norm}\widehat {\bf{S}}({\omega _j}), \forall j \in \mathcal{J}. 
\end{equation}
It can be observed that $\widehat{\bf{S}}({\omega _j}) $ has a unique non-zero element when ${\bar \theta _n} = {\theta _j}$.
Then, we collect $\widetilde {\bf{X}} ({\omega _j}), \forall j \in \mathcal{J}$ into a matrix form, given by
\begin{equation}
\begin{array}{cc}
\left[\widetilde {\bf{X}} ({\omega _1}),\cdots,\widetilde {\bf{X}} ({\omega _J}) \right]
={\bf{\Psi }}_{\rm norm}\left[\widehat {\bf{S}}({\omega _1}),\cdots,\widehat {\bf{S}}({\omega _J}) \right] \\
{~~~~~~={\bf{\Psi }}_{\rm norm}\left[\begin{array}{*{20}{c}}
0  & 0 & & 0\\
\vdots  & \vdots& & \widehat {\text{S}}_J({\omega _J})\\
\widehat {\text{S}}_1({\omega _1}) & \vdots & \cdots& \vdots\\
\vdots  & \widehat {\text{S}}_2({\omega _2})& & \vdots\\
0 & 0& & 0
\end{array}
\right]. }
\end{array}\!\!
\end{equation}
As such, we have
\begin{equation}
\begin{array}{cc}
\widetilde {\bf{X}} = \sum\limits_{j=1}^{J} {\bf{X}} ({\omega_j})
{={\bf{\Psi }}_{\rm norm} }  \widehat {\bf{S}}, 
\end{array}
\end{equation}
where $\widehat {\bf{S}} = [0,\cdots, {\text{S}}_J({\omega _J}), \cdots, {\text{S}}_1({\omega _1}),\cdots, {\text{S}}_2({\omega _2}),\cdots, 0]^T$.
Thus, we have completed the proof.
\end{proof}

\begin{remark}[Computational complexity]
    The computational complexity of the proposed RSV-SR method is dominated by two parts: 1) fast Fourier transform (FFT) on the received signal in place of DFT; 2) interior point method for solving the convex optimization in \eqref{e13}. 
    Based on Lemma \ref{lemma:linear DFT}, the complexity of the FFT is $\mathcal{O}(ML \log L)$. 
    For the convex optimization, using the CVX toolbox, the dominant complexity of the interior point method mainly involves each iteration and the number of iterations required to achieve precision $ \epsilon$, which is $\mathcal{O}\left( \sqrt{N}(MN^2+M^3)\log ({1}/{\epsilon}) \right)$. 
    Hence, the overall computational complexity of the proposed RSV-SR method is $\mathcal{O}\left(ML \log L + \sqrt{N}(MN^2+M^3)\log ({1}/{\epsilon}) \right)$, which can be simplified as $\mathcal{O}\left(MN^{2.5}\log ({1}/{\epsilon}) \right)$, since $N \gg M,\;N^{2.5}\gg L$. 
\end{remark}

\vspace{-12pt}
\section{Simulation Results}
In this section, numerical results are presented to evaluate the performance of the proposed DOA estimation method.

\subsection{System Setup and Benchmark Schemes}
The system setup is presented as follows.
A ULA with $M = 8$ antennas is employed, while $J=2$ coherent signal sources are assumed to impinge on the array.
The angular domain is uniformly discretized with a fine resolution of $0.2^\circ$ with $N = 900 $.
The amplitude errors are assumed to follow a Gaussian distribution with $g_m \sim \mathcal{N}(1,\; 0.1^2)$, and the phase errors follow $\varphi_m \sim \mathcal{N}(0,\; (10^\circ)^2), \forall m \in \mathcal{M}$.
The SNR of the received signal is defined by $10\text{lg}\left(\mathbb{E}\{|{\bm{s}}(t)|^2  \}/\sigma_n^2\right)\;\text{dB}$.
In addition, the performance metric is the root mean square error (RMSE), which is defined by
\begin{equation}
    {{\mathop{\rm RMSE}\nolimits}_{\rm DOA}} = \sqrt {\frac{{\sum\limits_{k = 1}^K {\sum\limits_{j = 1}^J {{{({{\widehat \theta }_{k,j}} - {\theta _j})}^2}} } }}{{KJ}}} ,
\label{e23}
\end{equation}
where $ K $ represents the number of Monte Carlo trials, which is set by $K = 500$.
In addition, ${\widehat \theta _{k,j}}$ represents the estimated value of the $ j $-th source's angle in the $ k $-th trial, while ${\theta_j}$ denotes the true value of the $j$-th source's angle.
The angular separation of two incident signals is denoted by $\Delta \theta$. If the differences of the two DOA estimates are both less than $\Delta \theta/2$, they are considered to be successfully identified. The resolution probability of DOA estimation is defined by
$(\text{CNT}_{\text{s}}/\text{CNT}_{\text{t}})\times 100\%$,
where $\text{CNT}_{\text{s}}$ denotes the number of successfully identified tests and $\text{CNT}_{\text{t}}$ represents the number of total tests \cite{lv2024dual}.
Moreover, the benchmark schemes used for performance comparison are demonstrated as follows.
\begin{itemize} 
    \item {\bf Maximum likelihood method:}
        This scheme minimizes the negative log-likelihood function derived from observed data under Gaussian noise assumptions \cite{stoica2002maximum}. It can resolve coherent signals because it does not assume signal independence, but is computationally intensive due to multidimensional optimization.
    \item {\bf Subspace fitting method:}
        For this scheme, subspace fitting method estimates signal parameters by fitting a model subspace to the observed data subspace, minimizing a cost function. 
        The optimal weighted subspace fitting (WSF) method, using $W_{opt} = \bar{A}^2 \Lambda_s^{-1}$, achieves asymptotic efficiency, particularly for highly correlated signals \cite{viberg1991sensor}.
\end{itemize}

\vspace{-8pt}
\subsection{Performance Analysis}

\begin{figure}[t]
	\centering
	\begin{minipage}[t]{0.475\linewidth}
		\centering
		\includegraphics[width=0.95\linewidth]{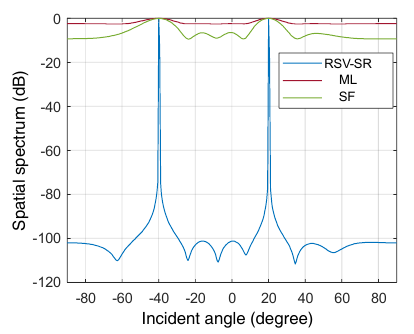}
        \captionsetup{font={footnotesize}}
		\caption{DOA spatial spectrums of the proposed RSV-SR method and benchmark schemes}
		\label{fig3:enter-label}
	\end{minipage}
    \hspace{0.002\textwidth} 
    \begin{minipage}[t]{0.475\linewidth}
		\centering
		\includegraphics[width=0.95\linewidth]{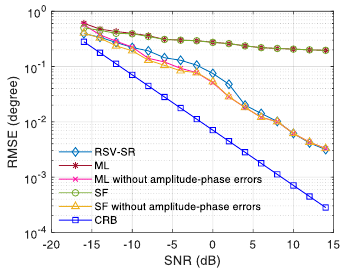}
        \captionsetup{font={footnotesize}}
		\caption{RMSE versus SNR where $L = 512$}
		\label{fig5:enter-label}
	\end{minipage}

    \vspace{8pt}
    \begin{minipage}{0.475\linewidth}
		\centering
		\includegraphics[width=0.95\linewidth]{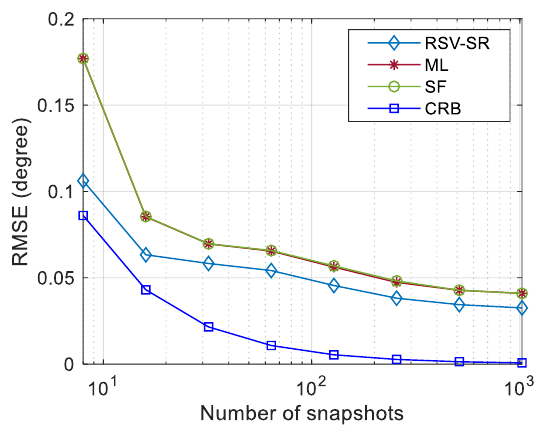}
        \captionsetup{font={footnotesize}}
		\caption{RMSE versus number of snapshots with SNR $10$ dB}
		\label{fig6:enter-label}
	\end{minipage}
    \hspace{0.002\textwidth} 
    \begin{minipage}{0.475\linewidth}
		\centering
		\includegraphics[width=0.95\linewidth]{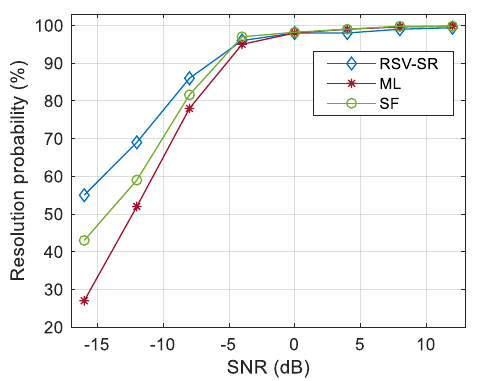}
        \captionsetup{font={footnotesize}}
		\caption{Resolution probability versus SNR}
		\label{fig7:enter-label}
	\end{minipage}
\vspace{-10pt}
\end{figure}

In Fig. \ref{fig3:enter-label}, we plot the spatial spectrums of DOA estimation of the proposed RSV-SR method, the ML method, and the SF method.
The two incident signal angles are set by $-40^ \circ $ and $20^ \circ $, with $L = 128$ snapshots.
In addition, the SNR is considered as $20$ dB.
It is observed from Fig. \ref{fig3:enter-label} that the proposed RSV-SR method exhibits sharper spectral peaks at the true target angles as compared to the two benchmark schemes.
This can be explained that the proposed method leverages the RSVs to construct a complete sparse reconstruction basis matrix, which effectively mitigates the amplitude-phase errors, thereby enabling higher DOA estimation accuracy for coherent sources.

We plot the RMSE versus SNR in Fig. \ref{fig5:enter-label}. We set incident signal angles as $-10^ \circ $ and $32^ \circ $ with $L = 512$ snapshots.
In addition, the RMSEs of ML and SF methods in an ideal scenario without amplitude-phase errors are denoted as ML without amplitude-phase errors and SF without amplitude-phase errors, respectively.
It is observed that the proposed RSV-SR method achieves better performance than the ML and SF methods.
Specifically, when the SNR exceeds $4$ dB, the RMSE of RSV-SR approaches that of ML without amplitude-phase errors and SF without amplitude-phase errors.
This is because the proposed RSV-SR method filter out most of the noise power when performing peak detection in the frequency domain, enhancing angle estimation accuracy. 

In Fig. \ref{fig6:enter-label}, we plot the RMSE versus number of snapshots. The incident signal angles are set as $-10^ \circ $ and $32^ \circ $, given $\text{SNR}=10\;\text{dB}$.
It is observed that the SRV-SR method achieves significantly better performance than the other benchmark schemes, especially when the number of snapshots is small.
This is because the proposed sparse reconstruction method uses the peak information in the frequency domain, and does not require a large number of time domain snapshots.

In Fig. \ref{fig7:enter-label}, we plot the resolution probability versus SNR.
We set incident signal angles as $15^ \circ $ and $20^ \circ $ with $L = 512$ snapshots. The angular separation $\Delta \theta = 5^\circ$.
It is observed that the RSV-SR method exhibits a significantly higher resolution probability at low-SNR regime.
Specifically, when the SNR is below $-8$ dB, the resolution probability of the RSV-SR method is $4.4\% $ higher than that of the SF method.
This can be intuitively understood, as when two sources are closely placed, the rank of the covariance matrix tends to be deficient.
However, the proposed RSV-SR method is independent of the covariance matrix of the received signal and only depends on the number of sampled angles. 

 \vspace{-10pt}
\section{Conclusions}
In this letter, we proposed a joint frequency-space sparse reconstruction DOA estimation method under coherent sources and amplitude-phase errors. 
First, we considered the array amplitude-phase errors and eliminated them by constructing the RSVs using an auxiliary source incident from known angles. 
Then, we collected the peaks of the spectrum of snapshot data as the received signals by leveraging the sparsity in the frequency domain.
Subsequently, by exploiting the sparsity in the spatial domain, we performed the DOA estimation using the sparse reconstruction method, which can effectively deal with coherent signals. 
Finally, numerical results were presented to demonstrate that the proposed RSV-SR DOA estimation method achieves higher accuracy and resolution while occupying less computational complexity as compared to various benchmark schemes in the low SNR regime and/or with a small number of snapshots.


\section*{Appendix A: Proof of Lemma \ref{lemma:linear DFT}}
The received signal of the $m$-th antenna at the $\ell$-th snapshot is given by
\begin{equation}
    [\bm{x}(\ell)]_m = \sum\limits_{j = 1}^J [{{\bm{b}}({\theta_j})}]_m {s_j}(\ell) + [{\bm{n}}(\ell)]_m\;,
\end{equation}
where $[{{\bm{b}}({\theta_j})}]_m$ denotes the $m$-th entry of the $j$-th steering vector ${\bm{b}}({\theta_j})$, while $[{\bm{n}}(\ell)]_m$ represents the $m$-th element of the noise vector ${\bm{n}}(\ell)$.
Then, by performing $L$-point DFT on $[\bm{x}(\ell)]_m,\ell = 1,2, \cdots, L$, we have
\begin{align}
    {\mathbf{X}_m}(\omega)& \!\!=\!\! \sum\limits_{\ell = 1}^{L} {\Big( {\sum\limits_{j = 1}^J [{{\bm{b}}({\theta_j})}]_m {s_j}(\ell) + [{\bm{n}}(\ell)]_m } \Big)} {e^{ - \jmath\frac{{2\pi }}{L}\omega\ell}} \\
    & \!\!=\!\! \sum\limits_{j = 1}^J [{{\bm{b}}({\theta_j})}]_m \sum\limits_{\ell = 1}^{L} {{s_j}(\ell){e^{ - \jmath\frac{{2\pi}}{L}\omega\ell}} \!\!+\!\! \sum\limits_{\ell = 1}^{L} {[{\bm{n}}(\ell)]_m{e^{ - \jmath\frac{{2\pi }}{L}\omega\ell}}} }\!, \notag
\end{align}
where ${\mathbf{X}_m}(\omega)$ is the DFT of $[\bm{x}(\ell)]_m$ at frequency $\omega$. Thus we have
\begin{equation}
    {\bf{X}} ({\omega}) = {\bf{B}}{\bf{S}}({\omega})+{\bf{N}}({\omega}),
\end{equation}
where ${\bf{S}}({\omega})$ and ${\bf{N}}({\omega})$ denote the DFT of the transmit signal $\bm{s}(t)$ and $\bm{n}(t)$, respectively.

\bibliographystyle{IEEEtran}
\bibliography{IEEE.bib}
\vfill
\end{document}